\documentclass[showpacs,amsmath,12pt,prd,onecolumn,nofootinbib]{revtex4}
\usepackage{graphicx} 
\usepackage{amsmath}
\usepackage{amssymb}
\usepackage{amsthm}
\usepackage{color}
\usepackage{tikz}
\usepackage{pgfplots}
\usepackage{tabularx} 
\usepackage{tikz}
\usetikzlibrary{arrows.meta,positioning,shadows}
\newtheorem{theorem}{Theorem}
\newtheorem{proposition}{Proposition}
\newtheorem{remark}{Remark}
\newtheorem{definition}{Definition}
\newtheorem{lemma}{Lemma}
\newtheorem{corollary}{Corollary}

\begin{document}
\title{Critical Coupling Surfaces in $\kappa(R,T)$ Gravity:\\
Regularity, Gravitational Screening, and Phase Transitions}

\author{Ginés R. Pérez Teruel}
\email{gines.landau@gmail.com}
\affiliation{Consellería de Educación, Cultura, Universidades y Empleo, Ministerio de Educación y Formación Profesional, Spain}

\begin{abstract}
We investigate the critical regime $\kappa(R,T)=0$ in $\kappa(R,T)$ gravity. While most studies assume a non--vanishing effective gravitational coupling, the existence of critical hypersurfaces where $\kappa$ vanishes is a generic feature of many admissible coupling functions. We show that the apparent singularity of the non--conservation equation is an artifact of a rewritten form of the conservation law and that the fundamental equations remain regular at $\kappa=0$. We further analyze the structure of critical hypersurfaces, derive the associated compatibility condition $(\nabla^\mu\kappa)T_{\mu\nu}=0$, and discuss their interpretation as gravitational screening surfaces separating attractive and repulsive gravitational phases. The existence of critical coupling hypersurfaces also obstructs a global Einstein--frame description, distinguishing $\kappa(R,T)$ gravity from theories based solely on algebraic redefinitions of the energy--momentum tensor. Possible cosmological and astrophysical consequences are briefly explored.
\end{abstract}

\maketitle

\section{Introduction}

One of the most important features of General Relativity (GR) is the universality of the gravitational coupling. In Einstein's theory, Newton's constant is assumed to be a fundamental constant of nature, independent of the local properties of spacetime and matter. Nevertheless, several modified theories of gravity have explored the possibility that the effective gravitational coupling may vary dynamically, either through additional scalar degrees of freedom, as in scalar--tensor theories \cite{BransDicke1961,FujiiMaeda2003,Faraoni2004}, through curvature-dependent modifications such as $f(R)$ gravity \cite{SotiriouFaraoni2010,DeFeliceTsujikawa2010}, or through explicit matter--geometry couplings as in $f(R,T)$ gravity \cite{Harko2011}. These approaches suggest that the strength of gravity may emerge as an effective quantity whose value depends on the local dynamical state of the gravitational and matter fields.

Among these alternatives, $\kappa(R,T)$ gravity was proposed as a phenomenological framework in which the Einstein gravitational coupling is promoted to an effective scalar function of the Ricci scalar $R$ and the trace $T=g^{\mu\nu}T_{\mu\nu}$
of the stress-energy tensor \cite{Teruel2018}. The field equations take the form
\begin{equation}
G_{\mu\nu}-\Lambda g_{\mu\nu}
=
\kappa(R,T)T_{\mu\nu},
\label{eq:field}
\end{equation}
where the effective coupling $\kappa(R,T)$ encodes possible interactions between geometry and matter beyond those present in GR. A distinctive feature of the theory is that the conservation law of the stress-energy tensor is modified. Taking the covariant divergence of Eq.~(\ref{eq:field}) and using the contracted Bianchi identities leads to
\begin{equation}
\nabla^\mu\!\left[\kappa(R,T)T_{\mu\nu}\right]=0,
\label{eq:fundamental_conservation_intro}
\end{equation}
which replaces the standard relation $\nabla^\mu T_{\mu\nu}=0$. Consequently, the effective quantity that remains conserved is the product $\kappa T_{\mu\nu}$ rather than the stress-energy tensor itself.

Although $\kappa(R,T)$ gravity was originally formulated at the level of field equations rather than from a variational principle, this situation is
not unprecedented in the history of theoretical physics. Maxwellian electrodynamics and GR were initially developed through their field equations before their modern action formulations became
established. The absence of a known action principle therefore does not by itself invalidate a gravitational theory, although the construction of a
consistent variational formulation remains an important open problem.\\

Since its original proposal, $\kappa(R,T)$ gravity has been applied to a variety of physical scenarios. Cosmological models exhibiting accelerated expansion have been investigated \cite{Ahmed2021,Dixit2022}, together with thermodynamic aspects of cosmic evolution \cite{Dixit2023}. The theory has also been used in the study of compact stars \cite{TeruelSinghRahaman2022, Taser2023}, wormholes \cite{Singh2024, Sarkar}, gravastars \cite{gravastar}, and exact space-time solutions. These investigations suggest that the framework is sufficiently flexible to accommodate a broad range of gravitational phenomena. Despite these developments, one conceptual and theoretical aspect of the theory has received little attention. Most studies implicitly assume that
$\kappa(R,T)\neq0$ throughout the space-time region under consideration. However, for a wide class of admissible coupling functions, the existence of critical hypersurfaces satisfying
$\kappa(R,T)=0$ is a generic possibility. At first sight, this condition appears problematic. Indeed, the non-conservation equation commonly used in the literature,
\begin{equation}
\nabla^\mu T_{\mu\nu}
=
-\frac{\nabla^\mu\kappa}{\kappa}
T_{\mu\nu},
\label{eq:naive_intro}
\end{equation}
seems to become singular when $\kappa\rightarrow0$. This observation naturally raises the question whether the hypersurface $\kappa=0$ represents a genuine singularity of the theory or merely a breakdown of a particular representation of the conservation law.

The purpose of the present work is to investigate this critical regime in a systematic way. We shall show that the apparent singularity of Eq.~(\ref{eq:naive_intro}) is not a singularity of the fundamental equations themselves. Instead, the condition $\kappa(R,T)=0$ defines a critical coupling surface with nontrivial mathematical and physical properties.

More specifically, we demonstrate that the field equations remain regular at $\kappa=0$, derive the compatibility condition that must be satisfied on regular critical hypersurfaces, and discuss the interpretation of these surfaces as gravitational screening boundaries separating attractive and repulsive gravitational phases. A related question concerns whether $\kappa(R,T)$ gravity can be globally reformulated as GR with an effective energy-momentum tensor. As we shall show, the existence of critical coupling hypersurfaces provides a fundamental obstruction to such a reformulation, distinguishing the theory from modified-gravity models, such as Rastall gravity \cite{Rastall1972}, that can be reduced to purely algebraic redefinitions of the matter sector \cite{Visser2018,Golovnev2024}.

The paper is organized as follows. In Sec.~II we introduce the notion of critical coupling surfaces and discuss representative examples. In Sec.~III we establish the regularity of the field equations at $\kappa=0$. Section IV derives the compatibility condition imposed by the conservation law on regular critical hypersurfaces. In Sec.~V we discuss cosmological and astrophysical applications of the critical-surface formalism. Finally, Sec.~VI summarizes our conclusions.

\section{Critical Coupling Surfaces}

The central object of the present work is the subset of spacetime where the effective gravitational coupling vanishes. Let $M$ denote the spacetime manifold and let $\kappa=\kappa(R,T)$
be a smooth scalar function of the Ricci scalar $R$ and the trace
$T=g^{\mu\nu}T_{\mu\nu} $ of the stress-energy tensor.

\begin{definition}
The \emph{critical coupling set} is defined as
\begin{equation}
\Sigma_c=
\left\{
p\in M:
\kappa(R,T)(p)=0
\right\}.
\label{SigmaDef}
\end{equation}
\end{definition}

In other words, $\Sigma_c$ is the set of all spacetime points at which the effective gravitational coupling vanishes. Since $\kappa$ is a smooth scalar field on spacetime, the condition $\kappa(R,T)=0$ defines a level set of $\kappa$.

The geometrical nature of $\Sigma_c$ depends on the behavior of the gradient of $\kappa$. A standard result from differential geometry states that level sets associated with regular values of a smooth function are smooth submanifolds \cite{Lee, Tu}. In the present case, this leads to the following result.

\begin{proposition}
Let $\kappa(R,T)$ be a smooth scalar function on a spacetime manifold $M$. If
\[
\nabla_\mu\kappa\neq0
\]
at every point of the critical set
\[
\Sigma_c=\{p\in M:\kappa(R,T)(p)=0\},
\]
then $\Sigma_c$ is a smooth hypersurface of spacetime.
\end{proposition}

\begin{proof}
The result follows directly from the regular level-set theorem \cite{Lee, Tu}. Since $\nabla_\mu\kappa\neq0$ on $\Sigma_c$, the value $0$ is a regular value of $\kappa$. Consequently, the level set
\[
\Sigma_c=\kappa^{-1}(0),
\]
is a smooth codimension--one submanifold of $M$.
\end{proof}
The condition $\nabla_\mu\kappa\neq0$
guarantees that the critical set possesses a well--defined normal vector
$n_\mu\propto\nabla_\mu\kappa$, allowing a local decomposition of spacetime into the regions
$\kappa>0$ and $\kappa<0$ separated by the hypersurface $\Sigma_c$. A point $p\in\Sigma_c$ satisfying
\begin{equation}
\nabla_\mu\kappa(p)\neq0,
\label{RegularPoint}
\end{equation}
will be called a \emph{regular point} of the critical set. If every point of $\Sigma_c$ satisfies (\ref{RegularPoint}), the critical coupling set will be said to be regular.

In four--dimensional spacetime, a regular critical set is therefore a three--dimensional hypersurface. Its normal direction is naturally determined by the gradient of $\kappa$. Whenever
\[
\nabla_\alpha\kappa\,\nabla^\alpha\kappa\neq0,
\]
we may introduce the normalized normal vector
\begin{equation}
n_\mu=
\frac{\nabla_\mu\kappa}
{\sqrt{\left|
\nabla_\alpha\kappa\,\nabla^\alpha\kappa
\right|}}.
\label{NormalVector}
\end{equation}

The vector $n_\mu$ is orthogonal to $\Sigma_c$ and points in the direction of maximal variation of the effective gravitational coupling. As we shall see in subsequent sections, it plays a fundamental role in the analysis of the conservation laws and in the characterization of transitions across critical coupling surfaces.
\begin{remark}[Degenerate critical points]
Throughout this work we restrict attention to regular critical surfaces,
namely those satisfying
\[
\kappa=0,
\qquad
\nabla_\mu\kappa\neq0.
\]
A different situation arises when both the coupling and its gradient vanish,
\[
\kappa=0,
\qquad
\nabla_\mu\kappa=0.
\]
In this case the regular level-set theorem no longer applies, and the
critical set need not be a smooth hypersurface.\end{remark} 
Such degenerate points may
be associated with extrema, saddle points, bifurcations, or topological
changes of the critical set $\Sigma_c$. Moreover, the compatibility
condition derived later becomes identically satisfied at these points.
Although the field equations remain regular, the local geometry of the
critical set can differ substantially from that of the regular critical
surfaces considered in the present work. The analysis of degenerate
critical points is left for future investigation.
\subsection{Representative examples}

The existence of critical coupling surfaces is not exceptional. On the contrary, it is a generic feature of many admissible choices of $\kappa(R,T)$. We now analyze some particular cases used in the literature for the running gravitational constant.

\subsubsection{Linear dependence on the trace}

Consider the simple coupling function

\begin{equation}
\kappa(T)=8\pi G-\lambda T,
\label{kappaT}
\end{equation}

where $G$ is Newton's gravitational constant and $\lambda >0$ is a coupling parameter.

The critical condition $\kappa(T)=0$, immediately yields

\begin{equation}
T=T_c=\frac{8\pi G}{\lambda}.
\label{Tc}
\end{equation}

Thus, the critical set is the level hypersurface defined by the constant value $T_c$ of the trace.

The gradient of the coupling function is

\begin{equation}
\nabla_\mu\kappa
=
-\lambda\nabla_\mu T,
\end{equation}

so that the regularity condition $\nabla_\mu\kappa\neq0 $ is equivalent to $\nabla_\mu T\neq0$. Consequently, the normal vector to the critical hypersurface may be written as

\begin{equation}
n_\mu=
\frac{\nabla_\mu T}
{\sqrt{\left|\nabla_\alpha T\nabla^\alpha T\right|}}.
\label{NormalT}
\end{equation}

Adopting the metric signature as $(+---)$, the energy--momentum tensor of a perfect fluid can be expressed as $T_{\mu\nu}
=
(\rho+p)u_\mu u_\nu
-pg_{\mu\nu}$, from which the trace $T$ follows directly: 

\begin{equation}
T=\rho-3p.
\label{PerfectFluidTrace}
\end{equation}

Therefore, the critical condition becomes

\begin{equation}
\rho-3p
=
\frac{8\pi G}{\lambda}.
\label{CriticalPerfectFluid}
\end{equation}

The critical hypersurface is thus determined by a specific balance between energy density and pressure. A particularly important case is that of a barotropic equation of state satisfying

\begin{equation}
p=w\rho,
\label{BarotropicEOS}
\end{equation}

where $w$ is a constant parameter. Substituting into (\ref{PerfectFluidTrace}) gives

\begin{equation}
T=(1-3w)\rho.
\end{equation}

Hence the critical condition (\ref{Tc}) becomes

\begin{equation}
\rho_c
=
\frac{8\pi G}
{\lambda(1-3w)},
\label{CriticalDensity}
\end{equation}

provided $w\neq 1/3$.

Several particular cases are worth mentioning. For dust ($w=0$), the critical density is $\rho_c=8\pi G/\lambda$, while for vacuum energy ($w=-1$) it becomes $\rho_c=2\pi G/\lambda$. By contrast, radiation ($w=1/3$) satisfies $T=0$ identically, preventing the formation of a critical coupling surface within the linear model (\ref{kappaT}). Thus, the existence of critical surfaces depends sensitively on the equation of state of the matter source.

More generally, the critical value $T=T_c$ partitions spacetime into regions where $\kappa$ is positive, negative, or vanishes. The hypersurface $T=T_c$ therefore constitutes a natural candidate for a phase boundary separating attractive and repulsive gravitational sectors.

\begin{figure}[h]
\centering
\includegraphics[width=0.65\textwidth]{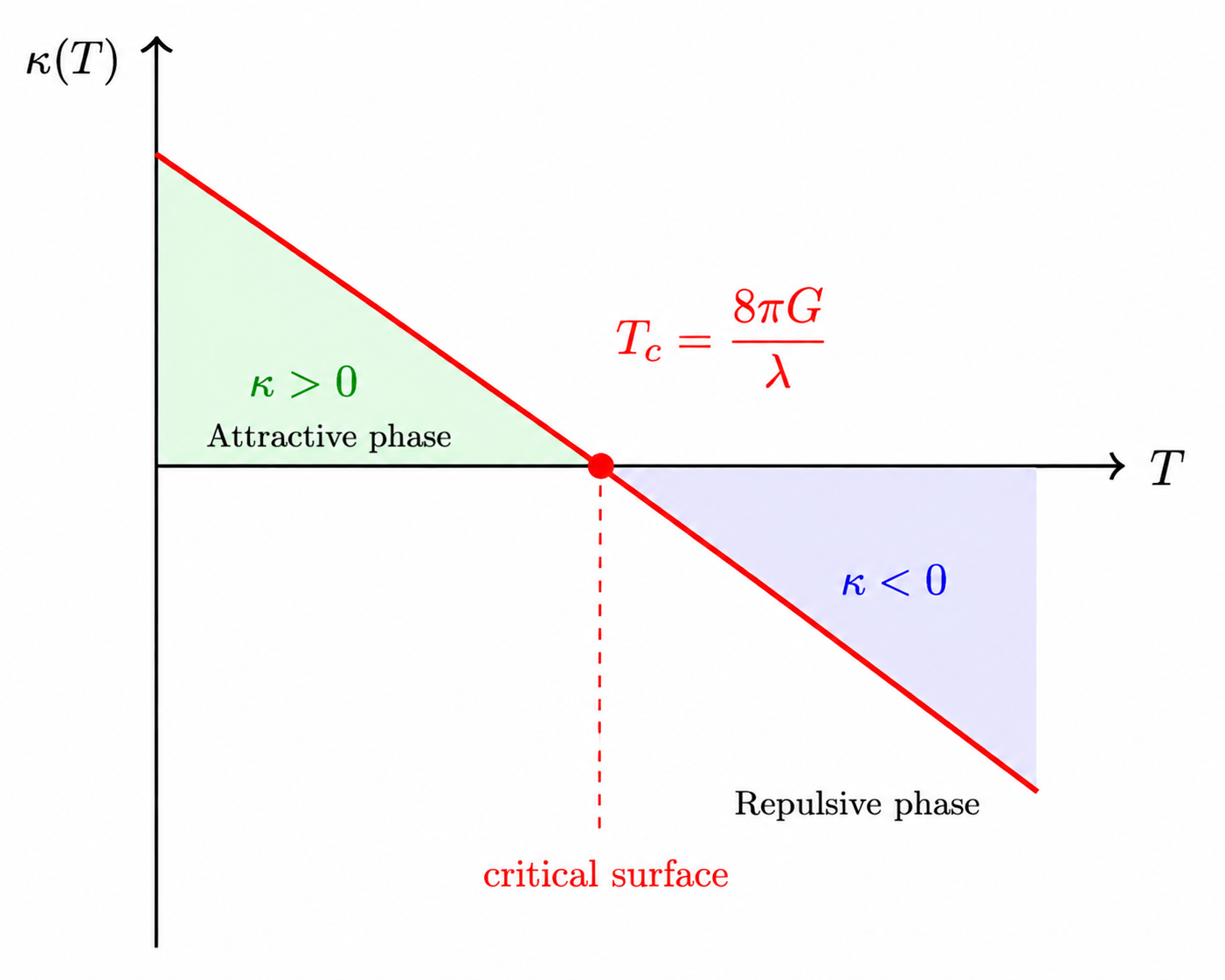}
\caption{Critical coupling surface for the linear model
$\kappa(T)=8\pi G-\lambda T$.
The hypersurface $T=T_c$ separates the attractive sector
($\kappa>0$) from the repulsive sector ($\kappa<0$).
}
\label{FigLinearTrace}
\end{figure}
\subsubsection{Linear dependence on the Ricci scalar}

As a second example, consider the coupling
\begin{equation}
\kappa(R)=8\pi G-\beta R,
\label{kappaR}
\end{equation}
where $\beta>0$ is a constant parameter. The critical condition $\kappa(R)=0$ yields
\begin{equation}
R=R_c=\frac{8\pi G}{\beta},
\label{Rc}
\end{equation}
so that the critical set is the hypersurface of constant scalar curvature $R_c$.

Since $\nabla_\mu\kappa=-\beta\nabla_\mu R$, the regularity condition $\nabla_\mu\kappa\neq0$ is equivalent to $\nabla_\mu R\neq0$, and the normal vector to the critical hypersurface takes the form
\begin{equation}
n_\mu=
\frac{\nabla_\mu R}
{\sqrt{\left|\nabla_\alpha R\nabla^\alpha R\right|}}.
\label{NormalR}
\end{equation}

Unlike the trace-driven model, the location of the critical surface is determined entirely by the geometry of spacetime. In particular, the critical regime is reached whenever the scalar curvature attains the value (\ref{Rc}), independently of the detailed properties of the matter source.

The choice $\beta>0$ is especially appealing from a physical perspective, since the effective gravitational coupling decreases as the curvature increases. The critical value $R_c$ therefore defines a curvature scale beyond which the coupling vanishes and subsequently changes sign. In this sense, regions of sufficiently large curvature may naturally enter a repulsive gravitational phase.

For Einstein spaces satisfying $R=4\Lambda$, the critical condition reduces to
\begin{equation}
\Lambda=\frac{2\pi G}{\beta},
\label{CriticalLambda}
\end{equation}
showing that the existence of a critical coupling surface is compatible with a positive cosmological constant. Moreover, the observed smallness of $\Lambda$ suggests that the corresponding critical curvature scale is reached only in highly curved regimes.

As before, the hypersurface $R=R_c$ separates regions with $\kappa>0$ and $\kappa<0$, and may therefore be interpreted as a transition surface between attractive and repulsive gravitational phases. In contrast to the previous example, the transition is now governed purely by the geometry of spacetime rather than by the matter content.

\subsubsection{Mixed curvature-matter coupling}

A particularly interesting example is obtained by allowing the effective gravitational coupling to depend simultaneously on curvature and matter through

\begin{equation}
\kappa(R,T)
=
8\pi G-\gamma RT,
\label{kappaRT}
\end{equation}

where $\gamma>0$ is a coupling parameter. The critical condition $\kappa(R,T)=0$ becomes

\begin{equation}
RT
=
\frac{8\pi G}{\gamma}.
\label{MixedCritical}
\end{equation}

Unlike the previous examples, the critical set is no longer determined by a constant value of either $R$ or $T$ separately, but by a nontrivial relation between geometry and matter. The critical coupling surface therefore emerges from the interplay between the curvature of spacetime and the local properties of the matter distribution.

The gradient of the coupling function is

\begin{equation}
\nabla_\mu\kappa
=
-\gamma
\left(
T\nabla_\mu R
+
R\nabla_\mu T
\right),
\label{GradientMixed}
\end{equation}

so that the regularity condition $\nabla_\mu\kappa\neq0$ requires that the combination
$T\nabla_\mu R+R\nabla_\mu T$ does not vanish identically on the critical surface. Whenever this condition is satisfied, the normal vector to the critical hypersurface is given by

\begin{equation}
n_\mu
=
\frac{
T\nabla_\mu R+R\nabla_\mu T
}
{
\sqrt{
\left|
(T\nabla_\alpha R+R\nabla_\alpha T)
(T\nabla^\alpha R+R\nabla^\alpha T)
\right|
}
}.
\label{MixedNormal}
\end{equation}

This expression illustrates a key difference with respect to the previous examples. The geometry of the critical surface is now controlled simultaneously by curvature gradients and matter gradients. Consequently, the transition between gravitational phases is driven neither by geometry alone nor by matter alone, but by their mutual interaction.

For positive $\gamma$, the critical condition (\ref{MixedCritical}) defines a hyperbola in the $(R,T)$ plane. The corresponding critical curve separates regions with positive and negative effective coupling, as shown schematically in Fig.~\ref{FigMixedCritical}.

\begin{figure}[h]
\centering
\includegraphics[width=0.65\textwidth]{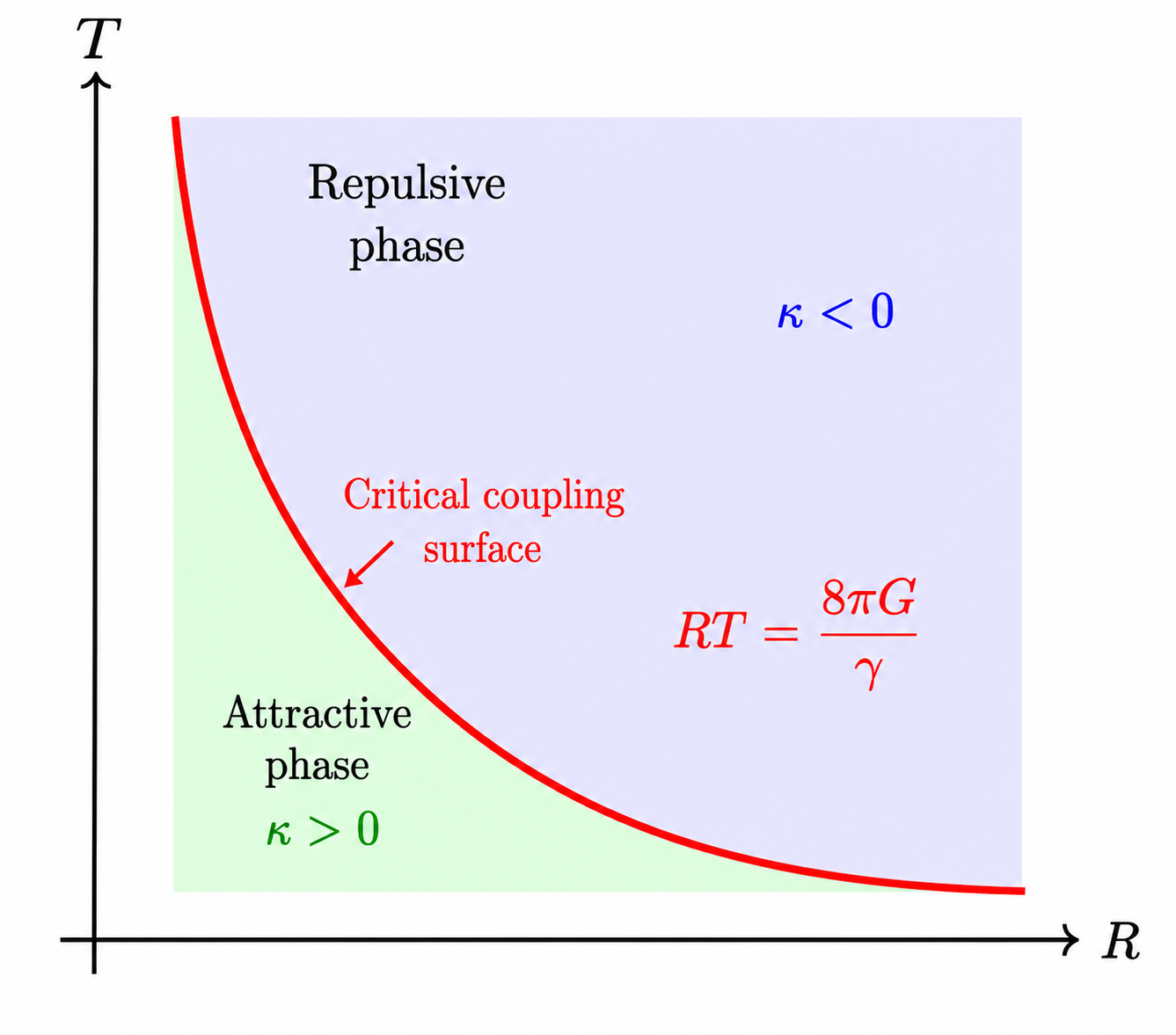}
\caption{Phase structure of the mixed coupling model
$\kappa(R,T)=8\pi G-\gamma RT$.
The critical curve $RT=8\pi G/\gamma$ separates the attractive sector
($\kappa>0$) from the repulsive sector ($\kappa<0$).
Unlike the purely curvature-driven or matter-driven models,
the transition is controlled by the combined effect of geometry and matter.
}
\label{FigMixedCritical}
\end{figure}

The mixed model possesses a richer phase structure than the purely trace-dependent or curvature-dependent cases. In particular, a transition between attractive and repulsive gravitational sectors may occur through variations in either the matter content, the spacetime curvature, or both simultaneously. This feature makes the coupling (\ref{kappaRT}) especially attractive for the study of high-density and high-curvature regimes, where geometry and matter are expected to interact strongly.

\subsection{Phase structure}

The existence of the critical set $\Sigma_c$ induces a natural decomposition of spacetime according to the sign of the effective gravitational coupling. We define
\begin{equation}
\begin{aligned}
\mathcal{M}_+ &=
\{p\in M:\kappa(R,T)(p)>0\},
\\
\Sigma_c &=
\{p\in M:\kappa(R,T)(p)=0\},
\\
\mathcal{M}_- &=
\{p\in M:\kappa(R,T)(p)<0\}.
\end{aligned}
\label{PhaseDecomposition}
\end{equation}

Thus, away from the critical surface, the spacetime is split into two open sectors, $\mathcal{M}_+$ and $\mathcal{M}_-$, separated by $\Sigma_c$ whenever the critical set is regular.

The field equations,

\begin{equation}
G_{\mu\nu}-\Lambda g_{\mu\nu}
=
\kappa(R,T)T_{\mu\nu},
\label{FieldAgain}
\end{equation}

show that the sign of $\kappa$ determines the sign of the effective coupling between matter and curvature. In the region $\mathcal{M}_+$, ordinary matter contributes to the curvature with the same sign as in General Relativity. By contrast, in $\mathcal{M}_-$ the matter contribution enters the field equations with the opposite sign, suggesting an effective repulsive gravitational sector.

The critical hypersurface $\Sigma_c$ therefore plays the role of a transition surface between attractive and repulsive gravitational phases. At $\Sigma_c$ itself, the matter contribution to the field equations is screened, and the geometry satisfies the effective vacuum equations

\begin{equation}
G_{\mu\nu}-\Lambda g_{\mu\nu}=0.
\end{equation}

This interpretation should be understood in a local and effective sense. The sign of $\kappa$ controls the sign of the matter source term in the field equations, but the dynamical admissibility of a transition across $\Sigma_c$ depends on the conservation law and on the regularity conditions derived below.

The mathematical consistency of this phase interpretation, as well as the behavior of the conservation laws near $\Sigma_c$, will be investigated in the following sections.
\section{Regularity of the Field Equations}

The first issue to address is whether the condition $\kappa(R,T)=0$ represents a genuine singularity of the theory. At first sight, this may seem plausible because the commonly used non-conservation equation contains the ratio $\nabla_\mu\kappa/\kappa$. However, the field equations themselves contain only the product $\kappa T_{\mu\nu}$. Therefore, the vanishing of $\kappa$ does not automatically produce a divergence.

In this section we show that, under standard smoothness assumptions, the critical set $\Sigma_c$ is not a singular set of the fundamental field equations. The result is local and should be understood as a statement about the regularity of the equations, not as a statement about the global existence or dynamical stability of solutions.

\begin{theorem}[Regularity at the critical coupling set]
Let $(M,g_{\mu\nu})$ be a smooth spacetime, let $T_{\mu\nu}$ be a smooth stress-energy tensor, and let $\kappa=\kappa(R,T)$ be a smooth scalar function. Consider the field equations
\begin{equation}
G_{\mu\nu}-\Lambda g_{\mu\nu}
=
\kappa(R,T)T_{\mu\nu}.
\label{RegularFieldEq}
\end{equation}
Let
\begin{equation}
\Sigma_c=\{p\in M:\kappa(R,T)(p)=0\}
\end{equation}
be the critical coupling set. If $g_{\mu\nu}$, $T_{\mu\nu}$, and $\kappa(R,T)$ are smooth at $p\in\Sigma_c$, then Eq.~(\ref{RegularFieldEq}) is algebraically regular at $p$. In particular, the right-hand side $\kappa T_{\mu\nu}$ is finite and smooth at $p$, and no divergence is generated solely by the condition $\kappa(p)=0$.
\end{theorem}

\begin{proof}
Since $\kappa(R,T)$ is a smooth scalar field and $T_{\mu\nu}$ is a smooth tensor field, their product $\kappa T_{\mu\nu}$ is a smooth tensor field. Therefore, at any point $p\in\Sigma_c$, where $\kappa(p)=0$, one simply has
\[
(\kappa T_{\mu\nu})(p)=0,
\]
provided $T_{\mu\nu}(p)$ is finite. Thus the matter source term in Eq.~(\ref{RegularFieldEq}) vanishes at $p$ but does not diverge.

The left-hand side of Eq.~(\ref{RegularFieldEq}) is constructed from the metric and its derivatives through the Einstein tensor $G_{\mu\nu}$, together with the smooth tensor $\Lambda g_{\mu\nu}$. Under the assumed smoothness of $g_{\mu\nu}$, this side is also finite wherever the curvature is finite. Hence the field equations remain algebraically well defined at $\kappa=0$.

Consequently, the vanishing of $\kappa$ alone is not a singularity of the field equations. Any singular behavior at $p$ would have to arise from a divergence of the metric, curvature, matter tensor, or their derivatives, not from the factor $\kappa(p)=0$ itself.
\end{proof}

The same conclusion holds for the fundamental conservation law. Taking the covariant divergence of Eq.~(\ref{RegularFieldEq}) and using the contracted Bianchi identity yields
\begin{equation}
\nabla^\mu(\kappa T_{\mu\nu})=0.
\label{FundamentalConservationRegularity}
\end{equation}
This equation contains no division by $\kappa$ and is therefore well defined at $\kappa=0$ whenever $\kappa$ and $T_{\mu\nu}$ are smooth.

\begin{proposition}[Regularity of the fundamental conservation law]
Under the hypotheses of the previous theorem, the conservation equation
\[
\nabla^\mu(\kappa T_{\mu\nu})=0
\]
is regular at every point of $\Sigma_c$.
\end{proposition}

\begin{proof}
Expanding the covariant derivative gives
\begin{equation}
\nabla^\mu(\kappa T_{\mu\nu})
=
(\nabla^\mu\kappa)T_{\mu\nu}
+
\kappa\nabla^\mu T_{\mu\nu}.
\label{ExpandedRegularity}
\end{equation}
If $\kappa$ and $T_{\mu\nu}$ are smooth, then both $(\nabla^\mu\kappa)T_{\mu\nu}$ and $\kappa\nabla^\mu T_{\mu\nu}$ are finite wherever the fields are smooth. At a point $p\in\Sigma_c$, the second term vanishes because $\kappa(p)=0$, while the first term remains finite provided $\nabla^\mu\kappa$ and $T_{\mu\nu}$ are finite. Hence Eq.~(\ref{FundamentalConservationRegularity}) is well defined at the critical set.
\end{proof}

The apparent singularity appears only after rewriting Eq.~(\ref{ExpandedRegularity}) in the form
\begin{equation}
\nabla^\mu T_{\mu\nu}
=
-\frac{\nabla^\mu\kappa}{\kappa}T_{\mu\nu}.
\label{SingularForm}
\end{equation}
This equation is obtained by dividing by $\kappa$ and is therefore valid only on the open set $M\setminus\Sigma_c$, where $\kappa\neq0$. It cannot be used as a fundamental equation on $\Sigma_c$ itself.

\begin{remark}
The results above do not imply that every solution crossing $\Sigma_c$ is dynamically regular. Curvature singularities, discontinuities, shocks, or distributional matter layers may still occur in particular solutions. The theorem only establishes that the condition $\kappa(R,T)=0$ is not by itself a structural singularity of the fundamental equations.
\end{remark}
The regularity of the field equations at $\kappa=0$ raises a natural question regarding the physical interpretation of the theory. Since the field equations can be formally rewritten as
\begin{equation}
G_{\mu\nu}-\Lambda g_{\mu\nu}
=
\widetilde T_{\mu\nu},
\qquad
\widetilde T_{\mu\nu}
=
\kappa(R,T)T_{\mu\nu},
\end{equation}
one may wonder whether $\kappa(R,T)$ gravity is merely a reformulation of General Relativity with an effective energy-momentum tensor, in a way analogous to the equivalence between Rastall gravity and Einstein gravity discussed by Visser and more recently emphasized by Golovnev \cite{Visser2018,Golovnev2024}. The following lemma shows that such an interpretation is, at best, valid only away from the critical set. The existence of critical coupling surfaces introduces a fundamental obstruction to a global Einstein--frame description.
\begin{lemma}[Breakdown of the effective Einstein--frame description]
Consider the theory
\begin{equation}
G_{\mu\nu}-\Lambda g_{\mu\nu}
=
\kappa(T)T_{\mu\nu},
\end{equation}
and define the effective tensor
\begin{equation}
\widetilde T_{\mu\nu}
=
\kappa(T)T_{\mu\nu}.
\end{equation}
On any open region where $\kappa(T)\neq0$, the field equations may be
rewritten in Einstein form
\begin{equation}
G_{\mu\nu}-\Lambda g_{\mu\nu}
=
\widetilde T_{\mu\nu}.
\end{equation}
However, this reformulation becomes singular on the critical set
\[
\Sigma_c=\{p\in M:\kappa(T)(p)=0\}.
\]
\end{lemma}

\begin{proof}
If $\kappa(T)\neq0$, the field equations can be rewritten algebraically as
\[
G_{\mu\nu}-\Lambda g_{\mu\nu}
=
\widetilde T_{\mu\nu}.
\]
On $\Sigma_c$, however,
\[
\widetilde T_{\mu\nu}
=
\kappa(T)T_{\mu\nu}
=
0
\]
for every finite tensor $T_{\mu\nu}$. Therefore, infinitely many distinct
matter configurations are mapped to the same effective source, and the
Einstein-form description becomes singular.
\end{proof}

\begin{remark}
The singular behavior at $\Sigma_c$ is unavoidable and independent of the
specific form of the coupling function. On the critical surface, the
effective tensor vanishes identically and the correspondence between
physical and effective sources collapses. Furthermore, for trace-dependent couplings such as
\[
\kappa(T)=8\pi G-\lambda T,
\]
the effective redefinition
\[
\widetilde T_{\mu\nu}
=
\kappa(T)T_{\mu\nu}
\]
may fail to be globally invertible even away from the critical surface.
Taking the trace gives
\[
\widetilde T
=
\kappa(T)T
=
(8\pi G-\lambda T)T
=
8\pi G\,T-\lambda T^2.
\]
Thus the effective trace is related to the physical trace through a
quadratic map. For a given value of $\widetilde T$, the equation
\[
\lambda T^2-8\pi G\,T+\widetilde T=0
\]
may admit two distinct solutions,
\[
T_1\neq T_2,
\]
satisfying
\[
\widetilde T(T_1)=\widetilde T(T_2).
\]

Therefore, the correspondence between physical and effective matter
sources is not necessarily one-to-one. Even when $\kappa\neq0$, the
Einstein-frame representation may only be locally invertible. 

This situation differs fundamentally from Rastall gravity, where a global
algebraic redefinition of the matter sector can be used to establish an
equivalence with General Relativity \cite{Visser2018,Golovnev2024}. In
the present theory, both the singular behavior at $\Sigma_c$ and the
possible non-bijectivity of the trace map obstruct such a global
identification. The critical surface therefore represents the strongest
manifestation of a more general loss of global equivalence between
$\kappa(T)$ gravity and an Einstein-frame description.
\end{remark}

Although the previous lemma was formulated for the subclass
$\kappa=\kappa(T)$, the same obstruction persists for the general theory
$\kappa=\kappa(R,T)$. Indeed, the effective tensor
\[
\widetilde T_{\mu\nu}
=
\kappa(R,T)T_{\mu\nu}
\]
becomes degenerate on the critical set
\[
\Sigma_c=\{\kappa(R,T)=0\},
\]
where the correspondence between physical and effective sources loses
invertibility. Moreover, because $\widetilde T_{\mu\nu}$ depends
explicitly on the curvature scalar $R$, the effective source generally
contains both matter and geometric contributions. Consequently, the
Einstein-frame interpretation is only local and does not define a
globally equivalent reformulation of the theory.

\section{Conservation Laws at the Critical Surface}

Having established the regularity of the field equations at $\Sigma_c$, we now investigate the constraints imposed by the conservation law on solutions that approach or cross the critical coupling surface.

The fundamental conservation equation follows from the contracted Bianchi identity and the field equations,
\begin{equation}
\nabla^\mu(\kappa T_{\mu\nu})=0.
\label{FundamentalConservation}
\end{equation}

Expanding the covariant derivative yields
\begin{equation}
(\nabla^\mu\kappa)T_{\mu\nu}
+
\kappa\nabla^\mu T_{\mu\nu}
=
0.
\label{ExpandedConservation}
\end{equation}

Away from the critical set, where $\kappa\neq0$, Eq.~(\ref{ExpandedConservation}) may be rewritten as
\begin{equation}
\nabla^\mu T_{\mu\nu}
=
-\frac{\nabla^\mu\kappa}{\kappa}T_{\mu\nu}.
\label{NonConservationLaw}
\end{equation}

However, as discussed in the previous section, Eq.~(\ref{NonConservationLaw}) ceases to be valid on $\Sigma_c$ because it is obtained by dividing by $\kappa$. The fundamental equation remains (\ref{ExpandedConservation}).

Let $p\in\Sigma_c$. Evaluating Eq.~(\ref{ExpandedConservation}) at $p$ and using $\kappa(p)=0$ immediately gives
\begin{equation}
(\nabla^\mu\kappa)T_{\mu\nu}=0.
\label{CriticalCondition}
\end{equation}

This relation represents the first nontrivial constraint imposed by the conservation law at the critical surface.

\begin{proposition}[Critical compatibility condition]
Let $\Sigma_c$ be a regular critical hypersurface, i.e. $\kappa=0$ and $\nabla_\mu\kappa\neq0$ on $\Sigma_c$. Then every smooth solution of the field equations satisfies
\begin{equation}
(\nabla^\mu\kappa)T_{\mu\nu}=0
\end{equation}
on $\Sigma_c$.
\end{proposition}

\begin{proof}
The result follows directly from Eq.~(\ref{ExpandedConservation}). Since $\kappa=0$ on $\Sigma_c$, the term $\kappa\nabla^\mu T_{\mu\nu}$ vanishes identically, leaving
\[
(\nabla^\mu\kappa)T_{\mu\nu}=0.
\]
No further assumptions are required beyond the smoothness of the fields.
\end{proof}

The geometric meaning of this condition becomes particularly transparent on regular critical hypersurfaces. Since $\nabla_\mu\kappa$ is normal to $\Sigma_c$, Eq.~(\ref{CriticalCondition}) implies that
\begin{equation}
n^\mu T_{\mu\nu}=0,
\label{NormalFluxCondition}
\end{equation}
where $n^\mu$ is the unit normal defined in Eq.~(\ref{NormalVector}).

Equation (\ref{NormalFluxCondition}) states that the stress-energy tensor possesses no component along the normal direction to the critical surface. Equivalently, the flux of stress-energy through $\Sigma_c$ must vanish.

Thus, not every solution of the field equations can cross a critical coupling surface smoothly. Regular crossings are possible only when the matter distribution satisfies the compatibility condition (\ref{CriticalCondition}).

This result shows that $\Sigma_c$ behaves neither as a singular boundary nor as a generic hypersurface. Instead, it acts as a constrained transition surface whose crossing is governed by a nontrivial matching condition inherited from the conservation law.

\section{Applications}

\subsection{FLRW cosmology and the critical density}

The critical-surface formalism acquires a particularly transparent interpretation in homogeneous and isotropic cosmology. We use the signature $(+---)$ and take the spatially flat FLRW metric in the form
\begin{equation}
ds^2=dt^2-a^2(t)d\vec{x}^{\,2}.
\end{equation}
As the matter sources, consider a comoving perfect fluid given by,
\begin{equation}
T_{\mu\nu}=(\rho+p)u_\mu u_\nu-pg_{\mu\nu},
\qquad
u^\mu=(1,0,0,0).
\end{equation}
Now we choose the trace-dependent model

\begin{equation}
\kappa(T)=8\pi G-\lambda T,
\label{FLRWKappa}
\end{equation}

and a barotropic equation of state

\begin{equation}
p=w\rho.
\label{BarotropicFLRW}
\end{equation}

Since

\begin{equation}
T=\rho-3p=(1-3w)\rho,
\end{equation}

the critical condition $\kappa(T)=0$ yields

\begin{equation}
\rho_c
=
\frac{8\pi G}
{\lambda(1-3w)},
\qquad
w\neq\frac13.
\label{CriticalDensityFLRW}
\end{equation}

This density coincides with the critical density previously identified in cosmological studies of the theory (see \cite{Teruel2018}). In the present framework, however, it acquires a new geometrical interpretation: $\rho_c$ is precisely the density at which the cosmological evolution reaches the critical coupling surface $\Sigma_c$.

The modified Friedmann equation can be written as

\begin{equation}
H^2
=
\frac{1}{3}\kappa(T)\rho
+
\frac{\Lambda}{3},
\label{FriedmannCritical}
\end{equation}

or equivalently,

\begin{equation}
H^2
=
\frac{8\pi G}{3}\rho
-\frac{\lambda}{3}(1-3w)\rho^2
+\frac{\Lambda}{3}.
\end{equation}

At the critical density, the matter contribution is completely screened and Eq.~(\ref{FriedmannCritical}) reduces to

\begin{equation}
H^2=\frac{\Lambda}{3}.
\end{equation}

Thus, the matter sector is locally screened from the gravitational field equations. For $\Lambda>0$, the dynamics is governed by an effective de Sitter phase.

The conservation law derived in the previous section imposes an additional and highly nontrivial restriction. Since

\begin{equation}
\nabla_\mu\kappa
=
-\lambda\nabla_\mu T,
\end{equation}

the compatibility condition

\begin{equation}
(\nabla^\mu\kappa)T_{\mu\nu}=0
\label{CosmoCompatibility}
\end{equation}

must hold on the critical surface.

\begin{theorem}[Invariance of the critical density]
Consider a spatially flat FLRW universe filled with a barotropic perfect fluid and governed by the coupling (\ref{FLRWKappa}). Let $\rho_c$ be given by (\ref{CriticalDensityFLRW}). Then every regular cosmological solution satisfies

\begin{equation}
\dot{\rho}=0
\end{equation}

whenever $\rho=\rho_c$.
\end{theorem}

\begin{proof}
For a spatially flat FLRW spacetime with signature $(+---)$,
the comoving four-velocity is $u^\mu=(1,0,0,0)$ and $
T_{00}=\rho.$
Since $T=T(t)$, one has
\[
\nabla_\mu\kappa
=
(-\lambda\dot T,0,0,0).
\]
Raising the index with the FLRW metric gives
\[
\nabla^0\kappa=-\lambda\dot T.
\]
The $\nu=0$ component of the compatibility condition
\[
(\nabla^\mu\kappa)T_{\mu\nu}=0
\]
therefore yields
\[
-\lambda \dot T\,\rho=0.
\]
For any physical cosmological solution with $\rho\neq0$, it follows that
\[
\dot T=0.
\]
Using $T=(1-3w)\rho$ and $w\neq1/3$, one obtains immediately
\[
\dot\rho=0.
\]
\end{proof}

The theorem shows that the critical density is not merely a distinguished value of the matter density. Rather, it defines an invariant critical surface of the cosmological dynamics. Regular FLRW solutions cannot cross $\Sigma_c$ transversely with nonzero $\dot\rho$.

A stronger statement follows by using the full conservation equation. For the FLRW geometry, the fundamental conservation law can be written as
\begin{equation}
\frac{d}{dt}\big(\kappa\rho\big)
+
3H\kappa(\rho+p)
=
0.
\label{FLRWConservation}
\end{equation}

For the barotropic model considered above, this becomes
\begin{equation}
\left(\kappa+\rho\frac{d\kappa}{d\rho}\right)\dot\rho
+
3H\kappa(1+w)\rho
=
0.
\label{DensityEvolution}
\end{equation}

Near the critical density, let
\begin{equation}
\delta=\rho-\rho_c.
\end{equation}
Since
\[
\kappa(\rho)
=
-\lambda(1-3w)\delta,
\]
Eq.~(\ref{DensityEvolution}) gives, to leading order,
\begin{equation}
\dot\delta
=
-3H(1+w)\delta
+
O(\delta^2).
\label{LinearizedCritical}
\end{equation}

Thus $\delta=0$ is an invariant solution of the local density dynamics. In particular,
\begin{equation}
\dot\rho\big|_{\rho=\rho_c}=0,
\qquad
\ddot\rho\big|_{\rho=\rho_c}=0,
\end{equation}
provided $H$ remains finite at the critical density.

For an expanding branch with $H>0$ and ordinary matter satisfying $w>-1$, Eq.~(\ref{LinearizedCritical}) shows that perturbations around the critical density decay exponentially. Therefore, the critical density is approached as a dynamical screening boundary rather than crossed transversely. This supports the interpretation of $\Sigma_c$ as an invariant critical surface separating the attractive and repulsive sectors.
The separatrix character of the critical density can also be understood from a global perspective. Rewriting the conservation law (\ref{FLRWConservation}) in the form

\begin{equation}
\frac{d}{dt}
\left[
a^{3(1+w)}
\kappa(\rho)\rho
\right]
=0,
\end{equation}

one obtains the first integral

\begin{equation}
a^{3(1+w)}
\kappa(\rho)\rho=
\mathcal C,
\label{IntegratedFLRWConservation}
\end{equation}

where $\mathcal C$ is determined by the initial conditions. Solutions reaching the critical density belong to the special branch
$\mathcal C=0.$ Generic cosmological solutions with $\mathcal C\neq0$
cannot reach the critical surface. This provides a global
interpretation of $\rho_c$ as a separatrix of the cosmological phase
space.

\begin{corollary}[Separatrix character of the critical density]
The critical density $\rho_c$ defines an invariant hypersurface of the cosmological phase space. Regular FLRW trajectories cannot cross $\rho=\rho_c$ transversely. Consequently, the critical density acts as a dynamical separatrix between the sectors $\kappa>0$ and $\kappa<0$.
\end{corollary}

\begin{proof}
The result follows directly from the invariance condition $\dot\rho=0$
at $\rho=\rho_c$ together with the local evolution equation (\ref{LinearizedCritical}).
\end{proof}

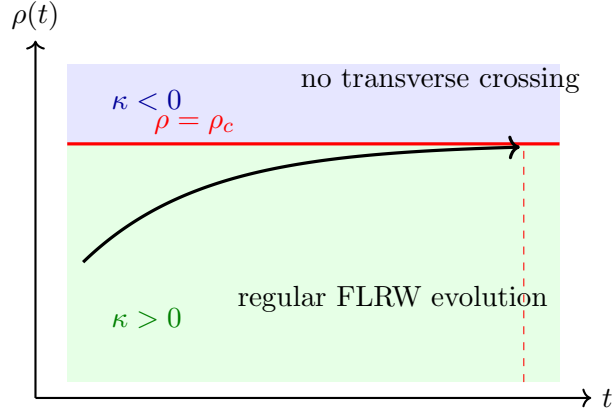
\begin{figure}[t]
\centering

\begin{tikzpicture}[scale=1.05]

\draw[->,thick] (0,0) -- (7,0) node[right] {$t$};
\draw[->,thick] (0,0) -- (0,4.5) node[above] {$\rho(t)$};

\fill[green!10] (0.4,0.2) rectangle (6.6,3.2);
\fill[blue!10] (0.4,3.2) rectangle (6.6,4.2);

\draw[very thick,red] (0.4,3.2) -- (6.6,3.2);

\draw[very thick,black,->]
plot[smooth,domain=0.6:6.1,samples=150]
(\x,{3.2 - 2.2*exp(-0.65*\x)});

\draw[dashed,red] (6.15,0.2) -- (6.15,3.2);

\node[red,above] at (2.0,3.2) {$\rho=\rho_c$};

\node[green!50!black] at (1.4,1.0) {$\kappa>0$};

\node[blue!60!black] at (1.4,3.75) {$\kappa<0$};

\node[align=center] at (4.5,1.25)
{\small regular FLRW evolution};

\node[align=center] at (5.1,4.0)
{\small no transverse crossing};

\end{tikzpicture}

\caption{
Schematic representation of the invariant cosmological critical density.
For the trace-dependent model, regular FLRW evolution approaches
$\rho_c$ asymptotically but cannot cross it transversely. The critical
density therefore acts as a dynamical screening boundary between the
$\kappa>0$ and $\kappa<0$ sectors.
}
\label{FigCriticalDensity}

\end{figure}

\subsection{Compact objects}

The compatibility condition derived in Sec.~4 also imposes nontrivial restrictions on compact astrophysical configurations. We continue to use the signature $(+---)$ and consider a static, spherically symmetric spacetime with line element

\begin{equation}
ds^2
=
e^{2\Phi(r)}dt^2
-
e^{2\Lambda(r)}dr^2
-
r^2d\Omega^2,
\end{equation}

filled again with a perfect fluid $T_{\mu\nu}=
(\rho+p)u_\mu u_\nu-p\,g_{\mu\nu}.$
For a static configuration, all matter variables depend only on the radial coordinate, and therefore $\kappa=\kappa(r)$. Consequently,

\begin{equation}
\nabla_\mu\kappa
=
\kappa'(r)\delta^r_{\mu}.
\end{equation}

Let $r=r_c$ be a regular critical surface satisfying

\begin{equation}
\kappa(r_c)=0,
\qquad
\kappa'(r_c)\neq0.
\end{equation}

The compatibility condition

\begin{equation}
(\nabla^\mu\kappa)T_{\mu\nu}=0
\end{equation}

then reduces to

\begin{equation}
T_{r\nu}\big|_{r=r_c}=0.
\end{equation}

For a perfect fluid at rest in the static frame, the only nonvanishing radial stress component is proportional to the pressure.

\begin{proposition}[Critical surfaces in perfect-fluid stars]
Let a static, spherically symmetric compact object be described by a perfect fluid with pressure $p(r)$. If $r=r_c$ is a regular critical coupling surface, then
\begin{equation}
p(r_c)=0.
\end{equation}
\end{proposition}

\begin{proof}
Since $\nabla_\mu\kappa$ is purely radial, the compatibility condition implies that the radial projection of the stress-energy tensor must vanish on the critical surface. For a perfect fluid, this projection is precisely the pressure. Hence $p(r_c)=0$.
\end{proof}

The converse statement does not generally hold. A vanishing pressure layer does not necessarily define a critical surface, since the additional condition
$\kappa(r_c)=0$ must also be satisfied.

This observation has an important physical consequence. In ordinary perfect-fluid stellar models, the pressure is strictly positive throughout the interior and vanishes only at the stellar boundary. Moreover, the density typically approaches zero near the surface. For the representative model
$\kappa(T)
= 8\pi G-\lambda(\rho-3p)$, one therefore finds
$\kappa
\approx
8\pi G $ near the stellar surface, so that the critical condition $\kappa=0$ is generally not satisfied.

Consequently, regular critical coupling surfaces appear to be absent in ordinary perfect-fluid stars. The compatibility condition strongly restricts their existence and prevents the formation of generic critical layers within regions of positive pressure.

The situation changes significantly for more general forms of matter. For anisotropic fluids,

\begin{equation}
T^\mu_{\ \nu}
=
\mathrm{diag}(\rho,-p_r,-p_t,-p_t),
\end{equation}

the compatibility condition admits a simple geometric interpretation.
Let
\begin{equation}
n_\mu=\nabla_\mu\kappa,
\end{equation}
denote the normal vector to a regular critical hypersurface. Then
\begin{equation}
(\nabla^\mu\kappa)T_{\mu\nu}=0,
\end{equation}
can be written as
\begin{equation}
T^\mu{}_\nu n^\nu=0.
\label{KernelCondition}
\end{equation}
Thus, the normal direction to the critical surface belongs to the kernel
of the stress-energy tensor. For a static anisotropic fluid with a radial
critical surface,
\[
n_\mu\propto\delta^r_\mu,
\]
Eq.~(\ref{KernelCondition}) reduces immediately to
\begin{equation}
p_r(r_c)=0,
\end{equation}
while the tangential pressure $p_t$ may remain nonzero.
Therefore, a regular critical layer is characterized by the vanishing of the principal stress in the direction normal to the hypersurface, while the density and tangential pressure may remain finite and nonzero. Internal critical layers are therefore not excluded in anisotropic configurations. This suggests that critical coupling surfaces may be more naturally associated with exotic compact objects, gravastars \cite{gravastar}, anisotropic stars \cite{Taser2023}, or other configurations possessing nontrivial internal structure.\\ 
An interesting observation in this context comes from recent studies of anisotropic compact stars in $\kappa(R,T)$ gravity \cite{Taser2023}. In particular, solutions constructed for the model $\kappa(T)=8\pi G-\lambda T$ exhibit nontrivial radial profiles of the effective coupling, including oscillatory behavior and stationary points near the stellar core. Although those configurations do not appear to reach the critical regime $\kappa=0$, the present analysis suggests a possible geometric interpretation of such features. Since the conservation law is governed by the combination $\nabla^\mu(\kappa T_{\mu\nu})=0$, variations of the coupling function are dynamically linked to the matter distribution. From this perspective, oscillations or extrema of $\kappa(r)$ may be viewed as precursors of critical-coupling behavior, reflecting the tendency of the system to approach configurations where the effective gravitational source $\kappa T_{\mu\nu}$ is strongly modified. This observation is particularly intriguing in anisotropic stars, where the compatibility condition reduces to $p_r(r_c)=0$ and therefore allows the existence of internal critical layers that are forbidden in ordinary perfect-fluid configurations.

Determining whether such solutions exist requires solving the corresponding equilibrium equations together with the critical matching condition, a problem that lies beyond the scope of the present work.
\begin{figure}[t]
\centering

\begin{tikzpicture}[scale=1.05,>=stealth]


\fill[green!10] (0,0) circle (2.4);
\draw[thick,black!70] (0,0) circle (2.4);

\fill[blue!10] (0,0) circle (1.15);
\draw[thick,blue!60!black] (0,0) circle (1.15);

\draw[very thick,red] (0,0) circle (1.55);

\node[green!50!black] at (0,2.75)
{$\kappa>0$};

\node[blue!60!black,align=center] at (0,-0.15)
{\small $\kappa<0$\\[-0.1em]
\small repulsive core};

\node[red,fill=white,inner sep=1pt]
at (0,2.02)
{$\Sigma_c:\ \kappa=0$};

\node at (0,-2.85)
{\small anisotropic compact object};

\draw[->,thick,black!70]
(-1.0,1.25) arc (128:78:1.55);

\node[black!70]
at (-1.05,1.78)
{$p_t$};

\draw[->,thick,black!70]
(0.15,0.2) -- (0.9,0.72);

\node[black!70]
at (1.45,0.92)
{$p_r$};

\draw[->,very thick,blue]
(1.55,0) -- (2.45,0);

\node[blue,right]
at (2.45,0)
{$n_\mu=\nabla_\mu\kappa$};

\draw[dashed,black!40]
(1.55,0.25) -- (3.25,1.35);

\draw[dashed,black!40]
(1.55,-0.25) -- (3.25,-1.35);


\begin{scope}[shift={(5.4,0)}]

\fill[green!10]
(-1.25,-1.55) rectangle (0,1.55);

\fill[blue!10]
(0,-1.55) rectangle (1.25,1.55);

\draw[very thick,red]
(0,-1.55) -- (0,1.55);

\node[red,above]
at (0,1.6)
{$\Sigma_c$};

\draw[->,very thick,blue]
(0,0) -- (1.55,0);

\node[blue,above]
at (1.05,0.08)
{$n_\mu$};

\node[green!50!black,align=center]
at (-0.72,0.55)
{\small $\kappa>0$\\[-0.1em]
\small attractive};

\node[blue!60!black,align=center]
at (0.75,-0.55)
{\small $\kappa<0$\\[-0.1em]
\small repulsive};

\node[
draw,
rounded corners,
fill=white,
inner sep=6pt,
align=center
]
at (0,-2.2)
{
\small $T^\mu{}_\nu n^\nu=0$\\[-0.15em]
\small $p_r(r_c)=0$
};

\end{scope}

\end{tikzpicture}

\caption{
Schematic representation of a regular critical layer in an anisotropic compact configuration. The critical hypersurface $\Sigma_c$ separates attractive ($\kappa>0$) and repulsive ($\kappa<0$) sectors. The compatibility condition implies that the normal vector $n_\mu=\nabla_\mu\kappa$ belongs to the kernel of the stress-energy tensor, $T^\mu{}_\nu n^\nu=0 $. For anisotropic matter, this yields
$p_r(r_c)=0$, providing a geometric characterization of regular critical layers.}
\label{FigCriticalLayer}
\end{figure}
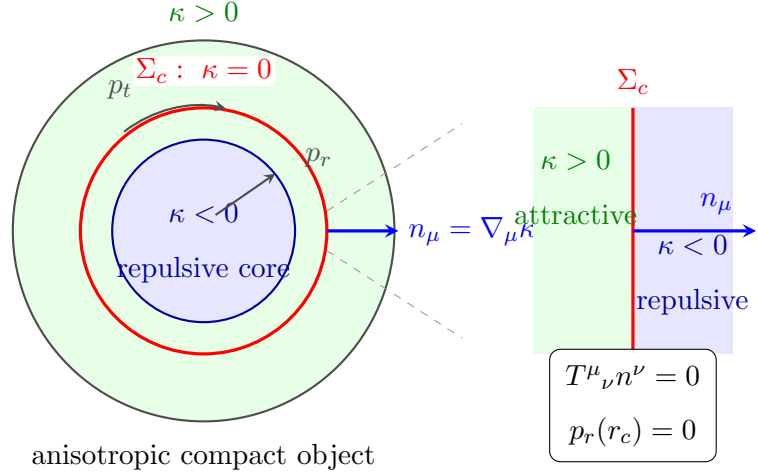

\section{Conclusions}

In this work we have investigated the critical regime
$\kappa(R,T)=0$ within $\kappa(R,T)$ gravity. We have shown that the apparent singularity commonly associated with the non-conservation law is not a singularity of the fundamental equations, and the underlying conservation equation remains regular and leads naturally to the compatibility condition
$(\nabla^\mu\kappa)T_{\mu\nu}=0$, which governs the behavior of matter on regular critical surfaces.

A central contribution of this work is that the set $\Sigma_c={p\in M:\kappa(R,T)(p)=0}$ can be interpreted as a genuine geometric object of the theory. Under suitable regularity assumptions, $\Sigma_c$ defines a smooth hypersurface separating regions with different effective gravitational behavior. This provides a natural geometric interpretation of the critical coupling regime and suggests that $\kappa(R,T)$ gravity possesses a richer internal structure than is apparent from analyses restricted to the sector $\kappa\neq0$.

An additional consequence of the present analysis concerns the interpretation of $\kappa(R,T)$ gravity itself. Although field equations may be locally rewritten in terms of an effective tensor $\widetilde T_{\mu\nu}=\kappa(R,T)T_{\mu\nu}$, the existence of critical hypersurfaces where $\kappa=0$ destroys the invertibility of this transformation. Consequently, unlike Rastall gravity, the theory does not admit a global reformulation as GR with a redefined energy-momentum tensor. The critical surfaces therefore represent genuinely new geometric structures rather than artifacts of a particular choice of variables.

The cosmological and astrophysical applications considered here illustrate the physical relevance of this point of view. In FLRW cosmology, the critical density becomes an invariant dynamical boundary that regular solutions cannot cross transversely. In compact objects, the compatibility condition strongly constrains the existence of critical layers and suggests that they may be more naturally realized in anisotropic configurations than in ordinary perfect-fluid stars.

Taken together, these results indicate that the geometry of the coupling function itself plays a nontrivial role in the theory. Rather than being merely a variable gravitational constant, $\kappa(R,T)$ appears to generate a stratified structure of spacetime composed of distinct coupling sectors separated by critical hypersurfaces. The study of such surfaces opens a new perspective on $\kappa(R,T)$ gravity and may provide further insight into its cosmological and astrophysical implications.


\begin{thebibliography}{99}

\bibitem{BransDicke1961}
C.~Brans and R.~H.~Dicke,
``Mach's Principle and a Relativistic Theory of Gravitation,''
Phys.\ Rev.\ \textbf{124}, 925--935 (1961).

\bibitem{FujiiMaeda2003}
Y.~Fujii and K.~Maeda,
``The Scalar--Tensor Theory of Gravitation'',
Cambridge University Press, Cambridge (2003).

\bibitem{Faraoni2004}
V.~Faraoni,
``Cosmology in Scalar--Tensor Gravity'',
Kluwer Academic Publishers, Dordrecht (2004).

\bibitem{SotiriouFaraoni2010}
T.~P.~Sotiriou and V.~Faraoni,
``$f(R)$ Theories of Gravity,''
Rev.\ Mod.\ Phys.\ \textbf{82}, 451--497 (2010).

\bibitem{DeFeliceTsujikawa2010}
A.~De Felice and S.~Tsujikawa,
``$f(R)$ Theories,''
Living Rev.\ Relativ.\ \textbf{13}, 3 (2010).

\bibitem{Harko2011}
T.~Harko, F.~S.~N.~Lobo, S.~Nojiri and S.~D.~Odintsov,
``$f(R,T)$ Gravity,''
Phys.\ Rev.\ D \textbf{84}, 024020 (2011).

\bibitem{Teruel2018}
G.~R.~P.~Teruel,
``$\kappa(R,T)$ gravity,''
Eur. Phys. J. C \textbf{78}, 660 (2018).

\bibitem{Ahmed2021}
N. Ahmed and A. Pradhan, ``Probing cosmic acceleration in $\kappa(R,T)$ gravity,''
Indian J Phys, \textbf{96}, 301 (2022)

\bibitem{Dixit2022}
A. Dixit, A. Pradhan, R. Chaubey
``Cosmological scenario in $\kappa(R,T)$ gravity,''
Int.\ J.\ Geom.\ Methods Mod.\ Phys. \textbf{19},2250013 (2022).


\bibitem{Dixit2023}
A. Dixit, S. Gupta, A. Pradhan, A. Beesham, ``Thermodynamics of the Acceleration of the Universe in the $\kappa(R,T)$ Gravity Model.'' Symmetry, \textbf{15}(2), 549 (2023).

\bibitem{TeruelSinghRahaman2022}
G.~R.~P.~Teruel,
K.~Newton Singh,
F.~Rahaman,
T.~Chowdhury,
``Possible existence of stable compact stars in
$\kappa(\mathcal R,\mathcal T)$ gravity,''
Int. \ J. \ Mod. \ Phys.\ A, \textbf{37}(31n32), 2250194 (2022)

\bibitem{Taser2023}
A.~Ta\c{s}er and F.~Do\u{g}ru,
``Compact stars in $\kappa(R,T)$ gravity,''
Astrophys.\ Space\ Sci. \textbf{368}, \ 6 
(2023).

\bibitem{Singh2024}
 K. N. Singh, G. R. P. Teruel, S. K. Maurya, T. Chowdhury and F. Rahaman, “Conservative wormholes
in generalized $\kappa(R, T)$- function”, JHEAp \textbf{44}, 132 (2024).

\bibitem{Sarkar}
S. Sarkar, N. Sarkar, F. Rahaman and Y. Aditya, “Wormholes in $\kappa (R, T)$ gravity”, arXiv:2207.12403
[gr-qc].

\bibitem{gravastar}
G.~R.~P.~Teruel,
K.~Newton Singh,
F.~Rahaman,
T.~Chowdhury, 
M.~Mondal,
``Testing of $\kappa(R,T)$-gravity through gravastar configurations,''
Phys. Dark Universe, \textbf{43}, 101404 (2024)

\bibitem{Rastall1972}
P.~Rastall,
``Generalization of the Einstein theory,''
Phys. Rev. D \textbf{6}, 3357 (1972).

\bibitem{Visser2018}
M.~Visser,
``Rastall gravity is equivalent to Einstein gravity,''
Phys.\ Lett.\ B \textbf{782}, 83--86 (2018).

\bibitem{Golovnev2024}
A.~Golovnev,
`` More on the fact that Rastall = GR'', Ann. Phys. \textbf{461} 169580. 



\bibitem{Lee}
J.~M.~Lee,
``Introduction to Smooth Manifolds,''
Springer, 2nd ed. (2012).

\bibitem{Tu}
L.~W.~Tu,
``An Introduction to Manifolds,''
Springer (2011).

\bibitem{FujiiMaeda}
Y.~Fujii and K.~Maeda,
``The Scalar-Tensor Theory of Gravitation,''
Cambridge University Press (2003).

\end{thebibliography}
\end{document}